\documentclass[reqno,11pt]{amsart}

\usepackage{hyperref}
\usepackage{graphicx}
\usepackage{slashed}
\usepackage{amscd}
\usepackage{tensor}
\usepackage{amsmath,amssymb,mathabx,bm}
\usepackage{esint}
\usepackage{dsfont}
\usepackage{enumitem}
\usepackage{accents}
\usepackage[off]{auto-pst-pdf} 
\usepackage{pst-grad} 
\usepackage{pst-plot} 
\usepackage[mathscr]{eucal}
\numberwithin{equation}{section}
\allowdisplaybreaks[1]

\title[Gravitation on the Light Cone]{Gravitation as a Statistical Theory on the Light Cone}
\author[J.M.\ Isidro]{Jos{\'e} M. Isidro}
\address{Instituto Universitario de Matem\'atica Pura y Aplicada \\ Universidad Polit\'ecnica de Valencia \\
Valencia 46022 \\ Spain}
\email{joissan@mat.upv.es}

\author[C.F.\ Paganini]{Claudio F. Paganini  }
\address{Fakult\"at f\"ur Mathematik \\ Universit\"at Regensburg \\ D-93040 Regensburg \\ Germany}
\email{claudio.paganini@ur.com}

\author[A.\ Pesci]{Alessandro Pesci }
\address{INFN\\ Sezione di Bologna\\ Via Irnerio 46\\ 40126 Bologna \\ITALY}
\email{pesci@bo.infn.it}

\newtheorem{Def}{Definition}[section]
\newtheorem{Thm}[Def]{Theorem}

\newtheorem{Lemma}[Def]{Lemma}
\newtheorem{Remark}[Def]{Remark}
\newtheorem{Corollary}[Def]{Corollary}

\newcommand{\Lcal}{{\mathcal{L}_x}}

\newcommand{\la}{\langle}
\newcommand{\ra}{\rangle}

\newcommand{\N}{\mathbb{N}}

\newcommand{\R}{{\mathord{\mathbb R}}}

\newcommand{\kbar}{\bar{k}}

\DeclareFontFamily{OT1}{rsfso}{}
\DeclareFontShape{OT1}{rsfso}{m}{n}{ <-7> rsfso5 <7-10> rsfso7 <10-> rsfso10}{}

\begin{document}
\begin{abstract} 
In this paper, we will explore Padmanabhan’s mesoscopic, statistical approach to
gravity \cite{padmanabhan2017atoms} with a twist. The general picture of his approach is that spacetime is made of large numbers of localized quantum degrees of freedom.  Padmanabhan assumed that the degrees of freedom 
of a given quantum state of geometry contribute, after averaging over fluctuations, 
a vector degree of freedom for spacetime at a point. 
For null vectors, this can be regarded as corresponding to one single vector, 
i.e. a pure state, for the statistical ensemble on the light cone at every point. 
In the present paper, we consider instead the case where the states of the
gravitational degrees of freedom are spread out and overlap, with only probabilistic
information on which of them determines the actual spacetime at a point. 
In the continuum limit, this corresponds to a mixed state for the statistical ensemble 
on the light cone at every point.\\
This change in assumptions leads to some interesting observations. When we define a statistical ensemble on the light cone, its variance ``knows'' about the interior of the light cone. As an intriguing consequence, we find that the cosmological constant can be related to the variance over the light cone.\\
With a mixed state, we can no longer derive the gravitational field equations from an entropy functional. Here, instead, we show that a naive implementation of the measure of a mixed state on the light cone in the variation principle leads to modified measure theories (MMT) as the grand canonical ensemble and allows one to reframe unimodular gravity as the canonical ensemble of a statistical theory on the light cone. 
\end{abstract}
\maketitle

\tableofcontents

\section{Introduction }
The present paper is a result of a series of papers comparing the structures and ideas of different approaches to fundamental physics~\cite{ethcfs, mmt-cfs, oct-cfs} as well as the upcoming articles~\cite{cfstdqg, tracedynamics}. The goal of this project is to motivate the community to establish an extensive collection of such articles as a sort of "Rosetta stone" for approaches to fundamental physics. The hope is that such a set of dictionaries of ideas helps the exchange across approaches and thereby catalyzes progress in the foundations of physics. The present article demonstrates the potential of this project, set in the context of thermodynamic approaches to gravity, it takes inspiration from Causal Fermion Systems (CFS) \cite{website, cfs} and links directly to Modified Measure Theories (MMT) \cite{NGVE,Fieltheory,manymeasures} and unimodular gravity (UG) \cite{UnimodularGRI,carballo2022unimodular,padilla2015note,smolin2009quantization,ng1991unimodular,jain2012testing}. 

The original idea of spacetime thermodynamics came about in the context of black holes. Bekenstein \cite{bekenstein1980black} realized that the black hole parameters satisfy the same relations as conjugate variables in thermodynamics and boldly suggested assigning an entropy to black holes proportional to the area of their event horizon. This idea was given further credibility by Hawking's discovery \cite{hawking1974black, hawking1975particle} that black holes indeed have a temperature, thus fixing the proportionality constant to be $\frac{k_B}{4 l_p^2}$. In his seminal paper in 1995 Jacobson \cite{jacobson1995thermodynamics} then showed that in fact the equations of General Relativity can be understood as thermodynamic balance equations across any null surface through any point in spacetime. To obtain the full Einstein equations in the bulk, one then needs to assume that the stress energy tensor of the matter fields is divergence-free. This fixes the dynamics in the bulk up to a constant of integration which can be identified with the cosmological constant.

Ever since, there has been great interest in this approach to the puzzle of (quantum) gravity. For recent work, see, for example, \cite{bekenstein1973black,bardeen1973four,hawking1975particle,gibbons1977action,jacobson1995thermodynamics,wald1999gravitation,chirco2010nonequilibrium,jacobson2016entanglement,bueno2017entanglement,parikh2018einstein,svesko2019entanglement,alonso2022thermodynamics,banihashemi2022thermodynamic,jacobson2023partition, requardt2023thermal}. For our results here, the paper \cite{alonso2022thermodynamics} is of particular interest, as they argue that rather than giving rise to the Einstein-Hilbert action, thermodynamic arguments actually favor UG which features equivalent dynamics. Our main focus, however, will be on the contributions by Padmanabhan \cite{padmanabhan2011entropy,padmanabhan2014general,padmanabhan2015one,padmanabhan2017cosmic,padmanabhan2017atoms} and his ideas of a mesoscopic Boltzmann approach to gravity. Hence in his reasoning, he assumes that there do exist internal degrees of freedom underlying every point of spacetime; however, he remains agnostic to the concrete nature of these microscopic degrees of freedom. The goal of his approach was to obtain a variational principle for gravity such that the equations of motion are invariant under $\mathcal{L}_M\rightarrow \mathcal{L}_M+C$ where $\mathcal{L}_M$ is the matter Lagrangian and $C$ is a constant.
To that end he studied the small sphere limit $x\rightarrow y$ of Synge's world function modified by adding a minimal length $\sigma_L(x,y)=\sigma(x,y) +L^2$ \footnote{Synge's world function $\sigma(x,y)$ is given by one half times the geodesic interval between $x$ and $y$ squared.} to derive a "gravitational density of states".  

His original derivation was performed in the Euclidean setting and carried over to the Lorentzian setting via analytic continuation. One of the authors of this paper later showed \cite{pesci2019quantum,pesci2019spacetime, pesci2020minimum, chakraborty2019raychaudhuri} that the same result can be obtained working directly in the Lorentzian setting. In the course of his investigation, Padmanabhan frequently identified the light cone as the configuration space for the internal degrees of freedom\footnote{The derivation for the density of states works equally well for time-like or space-like vectors, however the particular emphasis on null-vectors is due to their connection to horizons and thereby the original thermodynamic considerations.}. This is the starting point for our present work. In his later work \cite{padmanabhan2017atoms} he remarked that one should keep in mind that the null vectors that enter the variation principle for the entropy should be thought of as representing underlying gravitational quantum degrees of freedom which are localized in a small volume around a point $x$ with a sharply peaked distribution in momenta. He demands that when averaging over this state in momentum space, the result should be a null vector again. This implies that in the continuum limit where we consider a state on the light cone at a point $x$ this corresponds to a pure state, i.e. a Dirac delta in configuration space. 

In the present paper, we take inspiration from Causal Fermion Systems. In the continuum limit of the Minkowski vacuum, the underlying Hilbert space is given by the Dirac sea \cite{cfs}. It is therefore natural to associate these states with the gravitational sector of the theory. Hegerfeldt's theorem \cite{hegerfeldt1974remark,castrigiano2017dirac,finster2023incompatibility} tells us that vectors in this Hilbert space cannot be localized in any spatially compact region. In fact, if we consider the plane wave solutions, eigenstates of the momentum operator, the Dirac sea is built from maximally delocalized quantum states. 
As a result, we should consider that the volume at a point $x$ is not sourced by a single (pure) gravitational state, but by a collection of states;
and, assuming to have only probabilistic information on them,
by a mixed of state.
In the continuum limit, this corresponds to a general measure on configuration space, i.e., on the light cone.

 Padmanabhan \cite{padmanabhan2017atoms} argued that for the (pure) states he considered, the variance is small and hence can be neglected in a first-order approximation. This can not hold true if we consider mixed states. Working with mixed states leads to a number of interesting observations.
In particular, we find that taking the average over a probability field naturally gives rise to a time-like vector field that is similar in nature to the regularizing vector field that plays a key role in \cite{baryogenesis}. We also find that the variance can be linked to the cosmological constant if we consider the validity of the Einstein field equations as an observational fact.

Furthermore, it turns out that working with mixed states allows us to link several lines of research in the context of modified/thermodynamic gravity. In this paper, we explore the link with MMT and UG which is a subset of MMT, as explained in~\cite[Section 2.2]{mmt-cfs}. MMT was developed with the same goal in mind that was driving Padmanabhan's research: To find a variational principle for which the equations of motion
for the matter sector and the gravitational sector are both invariant,
to the addition of constants in the Lagrangian. Instead of focusing on the light cone, MMT replace the metric measure $\sqrt{-g}d^4x$ with a measure that is independent of the metric. As a result, additive constants to the Lagrangian drop out of the equations of motion as desired\footnote{We will touch briefly on the conjecture that these approaches are indeed two sides of the same coin if one considers the conformal factor of the metric as an independent degree of freedom. However, a rigorous discussion of this idea is beyond the scope of this paper.  }. 

In our present paper, we show that even a naive implementation of a mixed state on the light cone in the variation principle has interesting consequences. It naturally leads to MMT as the grand canonical ensemble and allows us to reframe UG as the canonical ensemble of a statistical theory on the light cone. Despite the ad hoc nature of these action principles, on the conceptual level they fit nicely with the considerations regarding baryogenesis \cite{baryogenesis} in the context of CFS \cite{cfs,website}. In particular, they give a complementary point of view on MMT. 
This perspective also fits with the arguments in \cite{alonso2022thermodynamics} as the derivation of UG requires the conventionally defined matter stress-energy tensor to be divergence-free, an assumption that is dropped in MMT. However, given the ad-hoc nature of the variation principles brought forward in this paper, they likely will not be the final answer, but at least they seem to be a good starting point for further investigations. 

Finally, we mention that this mathematical setup has the potential to link to Verlinde's heuristic derivation \cite{verlinde2011origin,verlinde2017emergent,hossenfelder2017covariant} of modified Newtonian dynamics from entropic considerations.

\subsection{Organization}
In Section \ref{sec:setup} we introduce basic mathematical notations and useful facts to keep in mind for the later constructions. In Section \ref{sec:paddy} we summarize the key arguments in Padmanabhan's papers. The mathematical core of the paper is in Section \ref{sec:probonlightcone} where we introduce the probability density on the light cone. Here we recall the fact that averaging over the light cone always results in a time-like vector except for pure states. Finally, we prove that the variance is a positive definite tensor. The core physical results are presented in the following two sections. In Section \ref{sec:physics} we observe that the variance can be connected to the cosmological constant. In Section \ref{sec:mmt} we show how theories with modified measures can be obtained if we consider a density of states instead of a probability measure on the light cone. Finally, in Section \ref{sec:outlook}, we give an overview of possible applications and generalizations of our mathematical framework.

\section{Mathematical Background}\label{sec:setup}
In the present work we will always consider $(M, g)$ to be a spacetime, i.e. a four-dimensional orientable Lorentzian manifold. The study of null-geodesics, i.e. geodesics $\gamma$ who's tangent vector $\dot\gamma^\mu$ satisfies 
\begin{equation}
    g_{\mu\nu}\dot\gamma^\mu\dot\gamma^\nu =0
\end{equation}
at every point has been of considerable interest as they play a crucial role in characterizing the causal structure of spacetime, see e.g., \cite{minguzzi2007causal,minguzzi2007causal2}.

From a physical perspective, null geodesics describe the path of light rays and are thus, in principle, amendable to experiment/observation. For that reason, they play a key role in the axiomatic reconstruction of spacetime from physical structures by Ehlers, Schild and Piran \cite{linnemann2021constructive,ehlers2012republication}. In the mathematical literature, there has been recent progress in the study of the geometry of the space of null-geodesics \cite{low1989geometry,hedicke2020conformally,chernov2010legendrian}.
In the present work, we are interested in a geometric structure closely associated with the space of null-geodesics, namely the bundle of past light cones. 
\begin{Def}[Past Light Cone Bundle (PLCB)] Let $(M,g)$ be a smooth spacetime, that is, a smooth time orientable Lorentzian manifold. Then the PLCB is the subset of the tangent bundle given by
    \begin{equation}
        \mathcal{L}:= \{ (x,v) \in T M| g_x(v,v)=0,  v\neq 0\ \text{ and }  v \text{ past directed}\}.
    \end{equation}
\end{Def}
The fiber $ \mathcal{L}_x$ of $\mathcal{L}$ at a point $x$ is just the past light cone at this point. In \cite{Bruce_2022} it was shown that the PCLB $\mathcal{L}$ is, in fact, a submanifold of the tangent bundle. This is a first sanity check that it makes sense to consider $\mathcal{L}$ as a configuration space for a physical theory. 

To make calculations more tractable, we will use a coordinate system based on an orthonormal tetrad $e_0^\mu,e_1^\mu,e_2^\mu,e_3^\mu $ where $e_0^\mu$ is time-like and the other tetrad vectors are space-like. Using the notation $\bar k \cdot \bar e ^\mu = k_1e_1^\mu+ k_2 e_2^\mu+k_3 e_3^\mu$ we notice that any null vector can then be written as 
\begin{equation}\label{eq:nullvectors}
    n^\mu = \lambda (e_0^\mu + \bar k \cdot \bar e ^\mu)
\end{equation}
where $\bar k = (k_1, k_2, k_3)$ is a three vector with unit Euclidean norm $|k|=1$. This makes the $\R\times S^2$ topology of the light cone explicit. We make use of this split in the definition of the past celestial sphere bundle.
\begin{Def}[Past Celestial Sphere Bundle (PCSB)]
    Let $(M,g)$ be a smooth spacetime and $X$ be a smooth, non-vanishing timelike vector field. Then the PCSB is the subset of the tangent bundle given by
    \begin{equation}
        CSM:= \{ v \in T_x M| g(v,v)=0, g(v,X)=1\}.
    \end{equation}
\end{Def}
From the results in \cite{Bruce_2022}, we get the following corollary.
\begin{Corollary}
    There is a global factorization of the PLCB by the PCSB. 
    \begin{equation}
        \mathcal{L}= \R_+ \times SM.
    \end{equation}
\end{Corollary}
This can be seen immediately from the definition of the PCSB and a global rescaling of the vector field $X$ by a scalar $\lambda$. This fact will become especially important in our follow-up paper \cite{qbit}, where we will remove most of the structures that we assumed for simplicity in the present paper. This concludes our collection of relevant notation and existing results.

\section{Padmanabhan's Thermodynamic Considerations} \label{sec:paddy}
In \cite{jacobson1995thermodynamics} Jacobson showed that the equations of motion of General Relativity on the light cone 
\begin{equation}\label{eq:jacobson}
 (R_{\mu\nu}-\kappa T_{\mu\nu})  n^\mu n^\nu  =0, \qquad \forall n^\mu \text{ with } g(n,n)=0
\end{equation}
can be recovered as a thermodynamic balance equation. To obtain the full Einstein equations from the equations on the light cone, one has to impose the fact that the stress-energy tensor is divergence free $\nabla^\mu T_{\mu\nu}=0$. This leads to 
 \begin{equation}\label{eq:einstein}
    R_{\mu\nu}-\frac{1}{2}R+\Lambda g_{\mu\nu}=  \kappa T_{\mu\nu}
\end{equation}
with $\Lambda$ as a free constant of integration.
Padmanabhan's idea was to extend on this argument and derive \eqref{eq:jacobson} from mesoscopic considerations. The philosophy of this approach is to assume that there exist microscopic degrees of freedom underlying spacetime and GR, which are macroscopic phenomena. However, no assumptions are made about the detailed properties of these microscopic degrees of freedom other than their existence. To this end Padmanabhan postulated, e.g. in \cite{padmanabhan2016atoms}, a density of microscopic degrees of freedom $\rho (x, \phi_A)$, where $\phi_A$ is a yet to be specified variable labeling the internal microscopic degrees of freedom. He arrived at this conclusion by investigating the minimum-length metric (also quantum metric, or q-metric for short), i.e., a matric description of spacetime with a built-in minimal length~\cite{kothawala2013minimal,kothawala2014entropy,stargen2015small,padmanabhan2016spacetime,nonlocal2023kothawala}.
The q-metric postulates that the squared integral between two points $\sigma^2(x,y)$ (i.e., two times Synge's world function \cite{synge1960relativity}) is modified on short scales\footnote{This allows for a covariant implementation of a minimal length scale in the universe. See~\cite{hossenfelder2013minimal} for a review of minimal length scenarios in quantum gravity and ~\cite{singh2021quantum} for a connection to holography.}. The simplest implementation of this idea is the addition of a constant to its square (the results derived from the q-metric do not depend on the details of the modification)
\begin{equation}
    \sigma(x,y)^2 \longrightarrow \sigma_L(x,y)^2=\sigma(x,y)^2+L^2 .
\end{equation}

 The q-metric corresponding to this modified world function is a bitensor, i.e. it is defined in terms of two points $x$ and $y$, which is singular everywhere in the coincidence limit (where $x\rightarrow y$). This singular behaviour is a direct consequence of the minimal length in the corresponding distance function. 

 One can calculate the volume and surface of an equi-geodesic ball with $\sigma_L(x,y)^2=C$. Taking the small-sphere limit, i.e. $x \rightarrow y$, one can compare the result from the modified distance with that from the original metric. In the ordinary spacetime metric volumes and areas vanish, of course, in the coincidence limit. For the q-metric, however, only the volume vanishes in the coincidence limit, while the area does not. This paves the way to introduce the density of gravitational degrees of freedom.  Padmanabhan then defines the density of microscopic degrees of freedom by the limit
\begin{equation}
    \rho (x, \phi_A)= \lim_{x\rightarrow y} \frac{\sqrt{h_{\sigma_L}}}{\sqrt{h_{\sigma_{L, \text{flat}}}}} = 1 - \frac{1}{6}L^2R_{\mu\nu}n^\mu n^\nu .
\end{equation}
where $\sqrt{h_{\sigma_L}}$ and and $\sqrt{h_{\sigma_{L, \text{flat}}}}$ are the area elements associated with the modified distance function in curved space and in flat space respectively and $n^\mu$ is the surface normal vector. Working in a Euclidean setting, Padmanabhan interpreted the zero point as the light cone after Wick rotation and $n^\mu$ in the limit accordingly as null vectors. This calculation led Padmanabhan to identify the internal degrees of freedom $\phi_A$ in his density of microscopic degrees of freedom by the set of null vectors $n^\mu$ at a point.

In his later papers \cite{padmanabhan2017atoms,padmanabhan2018kinetic,padmanabhan2022microscopic}, Padmanabhan then conceived a given quantum state of spacetime as the product of a collection of elementary quantum degrees of freedom over elemental volumes throughout all of spacetime. To go from the quantum state of spacetime to the classical description of spacetime he defined in~\cite{padmanabhan2017atoms} the average over quantum fluctuations at every point  $x$ representing an elemental volume through the functional integral 
\begin{equation}\label{eq:avg}
    \la n_\mu\ra= \int \mathcal{D}n \; n(x)_\mu P(n(x)^\mu, x) = l^\mu(x)
\end{equation}
where $P(n(x)^\mu, x)$ is parameterized by some null vector field $l^\mu(x)$, and $P$ is the probability that the actual quantum degree of freedom is given by $n(x)^\mu$ at $x$. Here, in principle, one should think of $P$  as a sharply peaked Gaussian in the variable $[n(x)^\mu- l^\mu(x)]$ supported in a small volume around the point $x$. This suggests, that $\la n_\mu\ra$ in \eqref{eq:avg} should be regarded as the expectation value of the momentum over a single quantum degree of freedom treated essentially as a pure. In the continuum limit, which we consider subsequently, such a state corresponds to a Dirac measure on the light cone.

Following Padmanabhan, we will treat the light cone itself as the configuration space of the internal variables $n^\mu$. However, inspired by CFS we now consider a
situation where the states of the underlying gravitational quantum degrees of freedom
overlap in spacetime,
with only probabilistic information of which state actually determines spacetime at $x$. In the present set-up this motivates the consideration of mixed states described by general invariant probability measures on the light cone $dP(n^\mu)$. Given \eqref{eq:nullvectors}, if the measure is absolutely continuous with respect to the Lebesgue measure of the tetrad coordinates, we can write this probability measure as $dP(n^\mu)=P(\lambda,\kbar)d\lambda d \Omega$, where $d\Omega$ is the canonical measure on $S^2$ and $d\lambda$ the canonical measure on $\R_+$. We define the average over the light cone at a point $x$ by
\begin{equation}
    \la n^\mu\ra_{\Lcal(P)}= \int_{R_+\times S^2} \lambda (e_0^\mu + \bar k \cdot \bar e ^\mu) P(\lambda, \bar k) d\lambda d\Omega .
\end{equation}
Following Padmanabhan we define the variance analogously.
\begin{Def}[Gravitational Dissipation]\label{def:gravdissipation}
 We define the gravitational dissipation as the variance over the light cone \begin{equation} \label{eq:variance}
    \sigma(P)^{\mu\nu} = \la n^\mu n^\nu\ra_{\Lcal(P)} - \la n^\mu\ra_{\Lcal(P)} \la n^\nu\ra_{\Lcal(P)}.
\end{equation}  
\end{Def}
Here $\la \cdot \ra_\Lcal$ again refers to the average over the configuration space of the internal variable $n^\mu$  (i.e. the light cone). In analogy to fluid mechanics, where
\begin{equation}
    \Sigma^{ab} = \la p^a p^b\ra_p - \la p^a\ra_p \la p^b\ra_p
\end{equation}
is the dissipation tensor and $\la \cdot \ra_p$ denotes the average over momentum space with $P^a=\la p^a\ra_p$ being the flow velocity of the macroscopic fluid. For the collection of (pure) states considered by Padmanabhan \cite{padmanabhan2017atoms} $\sigma^{\mu\nu}$ is a small correction originating from quantum gravity (as diffusion is a small effect compared to the overall flow of a fluid) and $l^\mu=\la n^\mu\ra  $ is a null vector depending on the single quantum degree of freedom at $x$ through the specific state $P(n(x)^\mu,x)$. Therefore a variation of the physical system with respect to $P(n(x)^\mu,x)$ is equivalent to a variation of the null vector $l^\mu$.
In Padmanabhan's approach, requiring extremization of total entropy (of spacetime and matter degree of freedom) with respect to $P(n(x)^\mu,x)$ hence leads to 
\begin{equation}\label{eq:paddysdream}
\delta [(R_{\mu\nu}-\kappa T_{\mu\nu}) \la n^\mu n^\nu\ra] \approx  \delta [(R_{\mu\nu}-\kappa T_{\mu\nu})  l^\mu l^\nu] =0,
\end{equation}
subject to the constraint $l^\mu l_\mu=0$. Where we used that $\delta P \sim \delta l^\mu$.
To leading order one thus gets that
\begin{equation} \label{eq:paddysmistake}
    (R_{\mu\nu}-\kappa T_{\mu\nu}) \la n^\mu n^\nu\ra   \approx (R_{\mu\nu}-\kappa T_{\mu\nu})\la n^\mu\ra \la n^\nu\ra =(R_{\mu\nu}-\kappa T_{\mu\nu}) l^\mu l^\nu =  0 .
\end{equation}
has to hold for all states $P$ and hence for all null vectors $l^\mu$. Demanding again, that the stress-energy tensor be divergence free $\nabla^\mu T_{\mu\nu}=0$, this leads to the Einstein field equations with a cosmological constant as a constant of integration, along the lines of the original argument by Jacobson.

As we shall see in the following, extracting expectation values of mixed states strongly impacts the argument just mentioned. Indeed, it is a well-known fact in Lorentzian geometry that the average of any two future (past) directed light-like vectors can only be null if the vectors are linearly dependent. This is the starting point for our results below. 

\section{A Probability Measure on the Light Cone} \label{sec:probonlightcone}
In this section we will assume two things: 
\begin{enumerate}
    \item The configuration space for the internal degrees of freedom for gravitation is the PLCB. 
    \item In analogy to Boltzmann's phase-space density $f(x,p^\mu)$ for fluids, we will assume a field of invariant Borel measures $dP(x,n^\mu)$ on the PLCB. If, in a fiber $\Lcal$ at a point $x$, the measure is absolutely continuous with respect to the Lebesgue measure of the tetrad coordinates, then it can be written in terms of the coordinates introduced above as $P(\lambda,\kbar)d\lambda d\Omega$, where $P(\lambda,\kbar)\geq0$ is nonnegative. 
\end{enumerate} 

\begin{Remark}
    In an abuse of notation, we will also write Dirac measures in terms of $P(\lambda, \kbar)$ with Dirac distributions. 
\end{Remark}

The goal of this section is to establish the properties of the statistical quantities required for Padmanabhan's formalism, in the context of mixed states. 
 For now, we ignore the spatial dependence and consider $dP(n^\mu)$ to be a probability measure on the past light cone at a point. 
This means that we assume 
\begin{equation}\label{eq:probability}
    \int_{R_+\times S^2}  P(\lambda, \bar k) d\lambda d\Omega=1 ,
\end{equation}
and for the variance to be finite  
\begin{equation}
    \int_{R_+\times S^2}  \lambda^2 P(\lambda, \bar k) d\lambda d\Omega< \infty.
\end{equation}
As a warmup, we show the following lemma which is a well-known fact in Lorentzian geometry.
\begin{Lemma}\label{lem:timelike}
   The vector $\la n^\mu\ra_\Lcal$\footnote{It is tempting to associate $-t^\mu=-\la n^\mu\ra_\Lcal$ with the flow of time, given it also shows up in the ``classical'' thermodynamic arguments \cite{rovelli2019travel}. Therefore, it would be tempting to identify $-t^\mu$ as the arrow of time. } is time-like except in the case where $P(\lambda, \bar k)=P(\lambda)\delta (\bar k -\bar k_0)$.
\end{Lemma}
\begin{proof}
For $P(\lambda, \bar k)=P(\lambda)\delta (\bar k -\bar k_0)$ it is clear that 
\begin{equation}
    \la n^\mu\ra_\Lcal= \int_{R_+} \lambda (e_0^\mu + \bar k_0 \cdot \bar e ^\mu) P(\lambda) d\lambda = (e_0^\mu + \bar k_0 \cdot \bar e ^\mu)\int_{R_+} \lambda P(\lambda) d\lambda.
\end{equation}
and therefore $\la n^\mu\ra_\Lcal$ is still a null vector.  Now we treat the case of a general $P(\lambda, \bar k)$.
For that, it is convenient to introduce the marginal probability density
\begin{equation}
    P(\lambda) = \int_{S^2} P(\lambda, \kbar) d\Omega.
\end{equation}
$P(\kbar)$ is then defined analogously. In addition, it is convenient to define 
\begin{equation}
    \lambda_{\text{avg}}= \int_{ R_+} \lambda P(\lambda) d \lambda.
\end{equation}
This allows us to write the expectation value over the light cone as
\begin{equation}
    \la n^\mu\ra_\Lcal= \lambda_{\text{avg}} \left(e_0^\mu + \int_{R_+ \times S^2} \frac{\lambda P(\lambda,\bar k) \bar k}{\lambda_{\text{avg}}} d\Omega d\lambda  \cdot \bar e ^\mu \right)
\end{equation}
Now it is clear that $\int_{ S^2} P(\lambda, \bar k) \bar k d\Omega$ is a convex combination over $S^2$ for more than one $\lambda$  and therefore that the inequality 
\begin{align}
    \left|\int_{\R_+\times S^2} \frac{\lambda P(\lambda,\bar k) \bar k}{\lambda_{\text{avg}}} d\Omega d\lambda \right|&\leq  \int_{\R_+} \frac{\lambda }{\lambda_{\text{avg}}} \left|\int_{S^2} P(\lambda,\bar k) \bar k d\Omega \right|d\lambda \\
    & \leq \int_{\R_+} \frac{\lambda }{\lambda_{\text{avg}}} \left(\int_{ S^2} P(\lambda, \bar k) |\bar k|d\Omega\right)d \lambda\\
    &=\int_{\R_+} \frac{\lambda }{\lambda_{\text{avg}}} P(\lambda)d\lambda =1,
\end{align}
where $|\cdot|$ is the Euclidean scalar product, in three dimensions is strict unless $P(\lambda, \kbar)$ is of the form $P(\lambda, \bar k)=P(\lambda)\delta (\bar k -\bar k_0)$. This finishes the argument.
\end{proof}
Note, that the exception of course includes the pure states considered by Padmanabhan with $P(\lambda, \kbar)=\delta(\lambda-\lambda_0)\delta(\kbar-\kbar_0)$. The following corollary will be useful for further calculations.  
\begin{Corollary} \label{cor:convenient}
For a probability measure over the light cone $dP(n^\mu)$ that is not of the form $P(\lambda, \bar k)=P(\lambda)\delta (\bar k -\bar k_0)$, we can choose an orthonormal basis $e_0^\mu,e_1^\mu,e_2^\mu,e_3^\mu $, that is, a coordinate system such that 
\begin{equation}
    \int_{ \R_+ \times S^2} \frac{\lambda P(\lambda, \bar k) \bar k}{\lambda_{avg}} d \lambda d\Omega=0
\end{equation}
and thus 
\begin{equation}
    \la n^\mu \ra_{\Lcal(P)}=\lambda_{\text{avg}} e_0^\mu.
\end{equation}
\end{Corollary}
In the next step, we will show that every state that is not of the form $P(\lambda, \bar k)=P(\lambda)\delta (\bar k -\bar k_0)$ can be normalized. 
\begin{Lemma}
    Let $dP(n^\mu)$ be a probability measure that is not of the form $P(\lambda, \bar k)=P(\lambda)\delta (\bar k -\bar k_0)$, and $e_0^\mu,e_1^\mu,e_2^\mu,e_3^\mu  $ a tetrad choosen according to Corollary \ref{cor:convenient}, then we can always find a related probability measure $d\Tilde{P}(n^\mu):= dP(\alpha n^\mu)$ such that $\la n^\mu\ra_\Lcal$ is a unit timelike vector.
\end{Lemma}
Let $\Tilde{P}(\lambda, \bar k)= P(\alpha\lambda, \bar k)$ with $\alpha>0$ then we have 
\begin{equation}
    \int_{R_+\times S^2} \lambda (e_0^\mu + \bar k \cdot \bar e ^\mu) \Tilde{P}(\lambda, \bar k) d\lambda d\Omega = \int_{R_+\times S^2} \lambda (e_0^\mu + \bar k \cdot \bar e ^\mu) P(\alpha\lambda, \bar k) d\lambda d\Omega \phantom{\frac{1}{\alpha^2}}
\end{equation}
Replacing $\frac{\Tilde \lambda}{\alpha} =  \lambda$ and $\frac{d\Tilde \lambda}{\alpha}= d\lambda$ we get
\begin{align}
    \phantom{\int_{R_+\times S^2} \lambda (e_0^\mu + \bar k \cdot \bar e ^\mu) \Tilde{P}(\lambda, \bar k) d\lambda d\Omega} & = \frac{1}{\alpha^2}\int_{R_+\times S^2} \Tilde\lambda (e_0^\mu + \bar k \cdot \bar e ^\mu) P(\Tilde\lambda, \bar k) d\Tilde\lambda d\Omega \\
    &= \frac{\lambda_{\text{avg}}}{\alpha^2}e_0^\mu
\end{align}
Setting $\alpha^2=\lambda_{\text{avg}}$ gives us 
\begin{equation}
    \la n^\mu\ra_{\Lcal(\Tilde{P})}= e_0^\mu.
\end{equation}
This finished the argument. 
\begin{Def}
 We will call a probability measure $d\Tilde{P}$  with $\la n^\mu\ra_{\Lcal(P)}= e_0^\mu$ a normalized state.   
\end{Def}
 The above lemma gives us a relation between the normalization of the resulting timelike vector field and the change in the underlying probability measure.

We are now ready to study the variance of a probability measure on the light cone. For that, we need the following definition.
\begin{Def}
    We call $P(\lambda,\bar k)$ degenerate of order $n$ if there exist $n$ linearly independent vectors $k_i\in S^2$ such that 
    \begin{equation}
        supp (P(\lambda, \bar k))\subset span\{k_i| i\leq n \in \N\}^\perp
    \end{equation}
    or if there are $n-1$ such vectors and \begin{equation}
        P(\lambda) = \delta (\lambda-\lambda_0).
    \end{equation}
\end{Def}

\begin{Thm}
    For any non-degenerate $P(\lambda,\bar k)$ 
    \begin{equation}
        \sigma(P)^{\mu\nu} = \la n^\mu n^\nu\ra_{\Lcal(P)} - \la n^\mu\ra_{\Lcal(P)} \la n^\nu\ra_{\Lcal(P)}
    \end{equation}
    is strictly positive.
\end{Thm}
\begin{proof}
We choose an orthonormal basis according to Corollary \ref{cor:convenient}.
Then we have the corresponding co-tangent basis $e^0_\mu,e^1_\mu,e^2_\mu,e^3_\mu$. Due to orthogonality for $i=\{1,2,3\}$ we have 
\begin{align}
    \left( \la n^\mu n^\nu\ra_{\Lcal(P)} \right.&- \left.\la n^\mu\ra_{\Lcal(P)} \la n^\nu\ra_{\Lcal(P)} \right) e^i_\mu e^i_\nu= \la n^\mu n^\nu\ra_{\Lcal(P)} e^i_\mu e^i_\nu \\
&= \int_{R_+\times S^2} \lambda^2 (e_0^\mu + \bar k \cdot \bar e ^\mu) (e_0^\nu + \bar k \cdot \bar e ^\nu) P(\lambda, \bar k) d\lambda d\Omega e^i_\mu e^i_\nu \\
&= \int_{R_+\times S^2} \lambda^2  ( k_i)^2 P(\lambda, \bar k) d\lambda d\Omega
\end{align}
This is positive by our assumption, that $P(\lambda, \bar k)$ be non-degenerate.
For $e_0$ we need a different argument. We calculate
\begin{align}
   \left( \la n^\mu n^\nu\ra_{\Lcal(P)} \right.&- \left.\la n^\mu\ra_{\Lcal(P)} \la n^\nu\ra_{\Lcal(P)} \right) e^0_\mu e^0_\nu=\int_{R_+}\lambda^2 P(\lambda)d\lambda-\left(\int_{R_+}\lambda P(\lambda)d\lambda\right)^2 \nonumber\\
   &= \int_{R_+} P(\kappa)d\kappa \int_{R_+}\lambda^2 P(\lambda)d\lambda - \left(\int_{R_+}\lambda P(\lambda)d\lambda\right)\left(\int_{R_+}\kappa P(\kappa)d\kappa\right)\nonumber \\
   &=\int_{R_+\times R_+}  \lambda (\lambda-\kappa) P(\lambda)P(\kappa) d\lambda d \kappa  \geq0
\end{align}
The positivity follows from Lemma \ref{lem:positivity} in Appendix \ref{sec:lemma} due to the fact that $P(\lambda)P(\kappa)$ is positive and symmetric. This concludes the proof.
\end{proof}


\section{Some Observations} \label{sec:physics}
In the following we apply the results of the previous section to the variation of entropy. Clearly, in general we can not apply straightforwardly Padmanabhan's extremization of entropy setting the variation of the state $P$ equal with the variation of the null vector $l^a$. Instead we always have to carry the variation with respect to the state $P$ along which leads to
\begin{align} \label{eq:variationP}
     0=\delta [(R_{\mu\nu} - \kappa T_{\mu\nu}) \langle n^\mu n^\nu \rangle] &=
\delta [(R_{\mu\nu} - \kappa T_{\mu\nu}) (\langle n^\mu \rangle \langle n^\nu \rangle + \sigma(P)^{\mu\nu}] \\
     & \Updownarrow \nonumber\\ \label{eq:interesting}
    \delta [(R_{\mu\nu} - \kappa T_{\mu\nu}) \sigma(P)^{\mu\nu}] &= -\delta[(R_{\mu\nu} - \kappa T_{\mu\nu}) \langle n^\mu \rangle \langle n^\nu \rangle]
\end{align}
at equilibrium. We now restrict to variations such that 
\begin{equation}
\delta[(R_{\mu\nu} - \kappa T_{\mu\nu}) \langle n^\mu \rangle \langle n^\nu \rangle] = (R_{\mu\nu} - \kappa T_{\mu\nu}) \delta[\langle n^\mu \rangle \langle n^\nu \rangle]    
\end{equation}
under the constraint $\langle n_\nu \rangle \langle n^\nu \rangle =-C$. Following \cite{pesci2019spacetime} this implies that the Einstein equation hold 
\begin{equation}
     R_{\mu\nu}-\frac{1}{2}R+\Lambda g_{\mu\nu}=  \kappa T_{\mu\nu}
\end{equation}
with the cosmological constant originating from the Lagrange multiplier. If we then formally integrate equation \eqref{eq:interesting} we get
\begin{equation}
    (R_{\mu\nu}-\kappa T_{\mu\nu}) \sigma(P)^{\mu\nu} = - \lambda_{\text{avg}}^2 (R_{\mu\nu}-\kappa T_{\mu\nu}) e_0^\mu e_0^\nu +G 
\end{equation}
with  $G=(R_{\mu\nu} - \kappa T_{\mu\nu}) \langle n^\mu n^\nu \rangle$ following from  \eqref{eq:variationP}, a constant independent of $P$. 
We now make use of the fact that the Einstein field equations hold for this system (with the trace term either on the geometric side or on the matter side) 
\begin{align}
    \label{eq:einstein2}
    R_{\mu\nu}-\kappa T_{\mu\nu}&= \left(\Lambda-\frac{\kappa}{2}T\right)g_{\mu\nu}\\
    R_{\mu\nu}-\kappa T_{\mu\nu}&= \left(\frac{1}{2}R-\Lambda\right)g_{\mu\nu}
\end{align}
to replace the relevant term on the right-hand side of \eqref{eq:interesting} we get the following intriguing result
\begin{equation} \label{eq:money}
    (R_{\mu\nu}-\kappa T_{\mu\nu}) \frac{\sigma(P)^{\mu\nu}}{\lambda_{\text{avg}}^2} = \Lambda-\frac{\kappa}{2}T + \left(\Lambda-\frac{\kappa}{2}T\right)g_{\mu\nu}\langle n^\mu n^\nu \rangle=\Lambda-\frac{\kappa}{2}T= \frac{1}{2}R-\Lambda
\end{equation}
where we used $g_{\mu\nu}\langle n^\mu n^\nu \rangle= \langle g_{\mu\nu} n^\mu n^\nu \rangle=0$ in the second step. Similarly we get
\begin{equation}\label{eq:moneynormalized}
    (R_{\mu\nu}-\kappa T_{\mu\nu})\sigma(\Tilde P)^{\mu\nu} = \Lambda-\frac{\kappa}{2}T =\frac{1}{2}R-\Lambda
\end{equation}
for normalized states. At the present state of our understanding we do not want to interpret too much into this relation but report it as something reassuring and promising. 
Nevertheless, it is tempting that the cosmological constant shows up either way, especially in light of the fact that so much of Padmanabhan's work was geared towards explaining the cosmological constant. 

As a final remark, it is interesting to observe that a mixed state gives rise to a vector field $t^\mu= \la n^\mu\ra_\mathcal{L}$ relevant for the computation similar to the regularizing vector fields that appear in Causal Fermion Systems in the context of baryogenesis \cite{baryogenesis}. In light of the results in \cite{jacobson2008einstein}, it is interesting to note that this vector field seems to have no dynamical relevance, but is more of a sort of a bookkeeping device.  In accordance with classical thermodynamic treatments, there is a plethora of states $P(\lambda, \kbar)$ that lead to the same right-hand side in \eqref{eq:moneynormalized}. However, as a bookkeeping device, it might be relevant in the context of Padmanabhan's earlier considerations.  In \cite{parattu2013structure, padmanabhan2009entropy} he derived GR in terms of thermodynamically conjugate variables, and in \cite{padmanabhan2012emergence,padmanabhan2014general,padmanabhan2010surface,padmanabhan2010equipartition,padmanabhan2004entropy} he derived what he calls the ``holographic equipartition'' for stationary spacetimes. All of these considerations feature the choice of a vector field. In Appendix \ref{sec:regularization} we calculated this vector field for several forms of $P(\lambda,\kbar)$ that are relevant in Causal Fermion Systems, and it turns out that $\lambda_{\text{avg}}$ is proportional to the inverse of the regularization length $\varepsilon$, which can be thought to be of the order of the Planck length. 

\section{Modified Measure Theories from a Density of States}\label{sec:mmt}
The Hawking-King-McCarthy and Malament theorem \cite{hawking1976new,malament1977class} states that under very weak causality conditions the causal order $(M, \prec)$ determines the metric $(M, g)$ up to a conformal factor. Formally, this implies that $(M, g)$ is equivalent to $(M, \prec)$ plus a volume form $\Phi$, or colloquially: spacetime is equal to causal order plus volume. The attentive reader might notice that, by working on the light cone, we already fixed the causal structure, so the only variable left is the conformal factor of the metric. If we drop the requirement typically invoked in thermodynamic derivations of gravity that the stress-energy tensor of matter needs to be divergence free $\nabla^\mu T_{\mu\nu}=0$, then this allows for MMT. 
 Instead of integrating the Lagrangian against the measure $\sqrt{-g}d^4x$ MMT promote the measure to be an independent quantity. The total action is then given by 
\begin{equation}
    S= \int_M L \Phi(A) d^4x \qquad \text{ with } \qquad L=\frac{-1}{\kappa} R(\Gamma, g) +L_m,
    \label{ActionwithPhi}
\end{equation}
where $\kappa$ is the gravitational coupling constant, the scalar curvature is given in terms of the connection and the metric, and the measure by 
\begin{equation}
    \Phi(A)= \frac{1}{6}\varepsilon^{\alpha\beta\mu\nu}\partial_\alpha A_{\beta\mu\nu},
\end{equation}
where $A_{\beta\mu\nu}$ is the tensor gauge potential of a non-singular exact $4$-form $\omega=dA$. Thus, the non-Riemannian volume element density $\Phi(A)$ is given by the scalar density of the dual field-strength associated with that potential. As a result, $\Phi(A) d^4x$ is invariant under general coordinate transformations. This introduces an extra scalar degree of freedom as can be easily seen by rewriting $\Phi(A)$ as $\chi(x)\sqrt{-g}$ where $\chi(x)=\frac{\Phi(A)}{\sqrt{-g}}$.

In the following, we will show that when working with a measure on the light cone, it is straightforward to implement it in such a way that it naturally leads to MMT. For that we will relax assumption \eqref{eq:probability} and reinstate the dependence of the state $dP(x,n^\mu)$ on the spacetime point $x$. In the following, we assume a global time-like vector field $t^\mu$ such that the condition from Corollary \ref{cor:convenient} is satisfied at every point. We then make the following definition. 
\begin{Def}[Density of gravitational states] Given a measure $dP(x,n^\mu)$ on the PLCB which can be represented in terms of coordinates by $P(x,\lambda,\bar k)d\lambda d\Omega \sqrt{|g|}d^4x$  the density of gravitational states is given by 
    \begin{equation}
        \chi(x)= \int_{R_+ \times S^2} P(x,\lambda,\bar k) d\Omega d \lambda. 
    \end{equation}
\end{Def}
While in the previous discussion we have fixed $\chi(x)=1$ we now allow it to take any positive value. This definition is in perfect analogy with the number density of particles in Boltzmann's treatment of fluids \footnote{   In this case we have $n(x)=\int_{\Omega(p^\mu)} f(x,p^\mu) d\omega$, where $p^\mu$ are the internal degrees of freedom of the microscopic degrees of freedom, $\Omega(p^\mu)$ is the domain of these internal degrees of freedom and $d\omega$ is a suitable measure on $\Omega(p^\mu)$. } and $\chi(x)$ counts the number of gravitational states that contribute to the volume density at a point $x$.

We are now ready to formulate an action on $\mathcal{L}$ as 
\begin{equation}\label{eq:unconstraint}
S[g_{\mu\nu}, L_m, P(x,\lambda, \kbar)]=\int_\mathcal{L} \left(  \frac{-1}{\kappa}R + L_m\right)dP(x,n^\mu)=\int_\mathcal{M} \left\la  \frac{-1}{\kappa}R + L_m\right\ra_\Lcal \sqrt{-g}d^4x
\end{equation}
which is to be minimized. A short calculation 
\begin{align}
 S[g_{\mu\nu}, L_m, P(x,\lambda, \kbar]    &= \int_\mathcal{L} \left(  \frac{-1}{\kappa}R + L_m\right) P(x,\lambda, \kbar) \sqrt{-g}d^4x d\lambda d\Omega\\
    &=\int_M d^4x \sqrt{-g} \left[\left( \frac{-1}{\kappa}R + L_m\right)\int_{R_+ \times S^2} P(x,\lambda,\bar k) d\Omega d \lambda\right]  \\
    &= \int_M \left( \frac{-1}{\kappa}R + L_m\right) \chi(x) \sqrt{-g}d^4x =  \int_M \left( \frac{-1}{\kappa}R + L_m\right)\Phi_A d^4x.  \label{eq:mmtaction}
\end{align}
shows that this gives rise to a spacetime action with a modified measure. 
Now what happens if we minimize \eqref{eq:mmtaction} with respect to  $P(x,\lambda, \kbar)$ (i.e. if we vary $P(x,\lambda, \kbar)$ itself)\footnote{Conceptually a variation of $P(x,\lambda, \kbar)$ represents a variation of the distribution of the momenta of the underlying gravitational states.}? The answer is straightforward
\begin{align}\label{eq:unconstraintvariation}
    \int_M &\left\la  \frac{-1}{\kappa}R + L_m\right\ra_{\Lcal(P+\delta P) }\sqrt{-g}d^4x \\
    &= \int_M d^4x \sqrt{-g} \left[\left( \frac{-1}{\kappa}R + L_m\right)\int_{R_+ \times S^2} P(x,\lambda,\bar k)+\delta P(x,\lambda,\bar k) d\Omega d \lambda\right]  \\
    &= \int_M \left( \frac{-1}{\kappa}R + L_m\right) [\chi(x)+\delta \chi(x)] \sqrt{-g}d^4x,  \label{eq:mmtactionvariation}
\end{align}
where $\delta \chi(x)= \int_{R_+ \times S^2} \delta P(x,\lambda,\bar k) d\Omega d \lambda$. Therefore, we see immediately that a variation of the state $P(x,\lambda, \kbar)$ induces a variation in the modified measure, and hence, if we vary the state under the requirement that the variation respect a volume constraint, as is usually required in MMT, the dynamics obtained from this variation principle are identical to that of MMT with a single modified measure. Following \cite{NGVE} we get the following system of equations
\begin{align}
    R_{\mu\nu}-\frac{1}{2}Rg_{\mu\nu}= \frac{\kappa}{2}\left(T_{\mu\nu}+ M g_{\mu\nu}\right) + \frac{1}{\chi}\left( \chi_{,\mu;\nu} -\right.&\left.g_{\mu\nu} \,\Box\,\chi  \right), \\
   \Box\,\chi- \frac{\kappa}{D-1}\left[ \left(M + \frac{1}{2} T\right)+ \frac{(D-2)}{2}L_m           \right]&\chi =0 .
\end{align}
Where $T$ is the trace of the stress energy tensor and the integration constant $M=-\frac{2 \Lambda}{\kappa}$ takes the role of the cosmological constant. We see that this system of equations can only give rise to solutions compatible with the equations of Einstein's General Relativity if $\frac{1}{\chi}\left( \chi_{,\mu;\nu} -g_{\mu\nu} \Box\chi  \right)=0$. This holds, e.g., if $\chi=\text{const.}$ which requires $\left(M + \frac{1}{2} T\right)+ \frac{(D-2)}{2}L_m    =0$. That is, for the right hand side of \eqref{eq:variance} to vanish and for $L_m=0$, which can be satisfied if we minimize the action with respect to the fields in $L_m$. An interesting consequence of the above equations of motion is that we get a non-conservation of the conventionally defined matter stress energy that does not include a contribution from the field $\chi$, 
\begin{equation}\label{eq:nonconservation}
   \nabla^\mu T_{\mu\nu}= -2\, \frac{\partial L_m}{\partial g^{\mu\nu}}\,g^{\mu\alpha}\nabla_\alpha \ln{\chi}.
\end{equation}
This explains why the conventional derivation of the Einstein field equations from thermodynamic considerations never considered MMT as a possible alternative solution compatible with the usual arguments on the light cone. 

Now what happens if instead of general variations of the state $P(x, \lambda, \kbar)$, we constrain the variation to a subclass of states with $\chi(x)=\text{const.}$? This leads us directly to unimodular gravity. 

\subsection{Unimodular Gravity from a Density of States.}
Unimodular gravity was originally developed with a goal similar to MMT. If one limits the class of admissible variations of the metric to those that leave $\sqrt{-g}$ invariant, then adding a constant to the Lagrangian will not change the gravitational equations of motion. UG has been shown to appear as a special case of MMT with two measures when formulated in a generally covariant form~\cite{UnimodularGRI} 
\[ \mathcal{S}= \int d^4x \,\sqrt{-g} \, (R + 2 \Lambda + \mathcal{L}_m) - \int d^4 x \, \Phi(A) \,2 \Lambda \:. \]
 Here, a priori, $\Lambda$ is a dynamical scalar field, and~$\Phi(A)$ is as above.
However, variations with respect to~$A$ imply that~$\Lambda$ is a constant, whereas variations with respect to~$\Lambda$ yield~$\Phi(A) =\sqrt{-g}$.

This variation principle can be reformulated in terms of a measure on the PLCB as a minimization of 
\begin{equation}\label{eq:unimodular}
    \mathcal{S}= \int d^4x \sqrt{-g} \left\la(R + 2 \Lambda + \mathcal{L}_m)\right\ra_\Lcal
\end{equation}
under the constraint 
\begin{equation}\label{eq:constraint}
    \int d^4x \sqrt{-g} \left\la\Lambda \right\ra_\Lcal= const.
\end{equation}
If we vary the constraint with respect to $P(x,\lambda,\kbar)$ then we get as above, that ~$\Lambda$ is a constant, while variations with respect to~$\Lambda$ yields that $\chi(x)$ must be constant. All other variations proceed as usual. Hence, the constraint restricts us to the subclass of measures $dP(x, n^\mu)$ for which the density of states is constant in spacetime. Note that from a thermodynamic perspective, the constraint \eqref{eq:constraint} now has a clear interpretation: it is nothing else than the average energy $\la E\ra$ per state. This is again remeniscent of Padmanabhan's holographic equipartition~\cite{padmanabhan2010equipartition} and allows us to interpret the cosmological constant as the average energy per gravitational degree of freedom and the constraint as Boltzmann's equipartition of energy per degree of freedom\footnote{Alternatively one could try to  interpret it as the average spacetime volume per gravitational degree of freedom.}. 

At this point it is worth noting again that none of the variation principles in this section depend on the time-like vector field that we can associate with $dP(x, n^\mu)$ and accordingly this is truly just a bookkeeping device and we need not worry about the phenomenological constraints discussed in \cite{jacobson2008einstein}. We would also remind the reader that the physically relevant quantity derived from $dP(x, n^\mu)$ is Lorentz invariant and that, in general, there is a huge number of measures $dP(x, n^\mu)$ with the same $\chi(x)$ and $\la n^\mu\ra_\Lcal$.

\section{Summary and Outlook}\label{sec:outlook}
For better readability, we split this section into three parts. A short summary, an outlook that ranges from concrete to speculative, and the conclusion. 

\subsection{Summary}
On the mathematical side, starting from an invariant probability measure on the light cone, we presented in full detail the well-known fact from Lorentzian geometry that an average over the light cone always leads to a time-like vector field. Furthermore, we showed that the variance associated with any nondegenerate measure is a positive two tensor. 

On the physical side, we showed that the same setup can give rise to modified measures in a straightforward manner. Furthermore, we showed that unimodular gravity (UG), which is a subset of modified measure theories (MMT), can be characterized in terms of thermodynamic concepts. While MMT corresponds to the grand canonical ensemble where the number of gravitational degrees of freedom is free to change, UG corresponds to the canonical ensemble where the density of gravitational states is constant throughout spacetime. This aligns neatly with the observation in \cite{baryogenesis} in the context of causal fermion systems (CFS) that the matter/antimatter asymmetry originated in a shift of states from the gravitational sector to the matter sector. The present work strengthens the case for modified measures playing a key role in this mechanism. The fact that the stress energy tensor is not conserved in MMT is, in this context, not a bug but an essential feature. 

\subsection{Outlook}
In the present work, we established the basic viability of a theory assuming a field of probability measure or a density of states. We will now show that this rather simple idea actually opens a window to a wide variety of possibilities for future avenues of research, some more mathematical in nature, some more physical.

First, as noted in the body of the text, working on the light cone a priori assumes a lot of structure, as the light cone fixes the metric except for the conformal factor. If we intend for the formalism to have any bearing on the problem of quantum gravity, this is clearly not satisfactory. For this we would want to be able to derive the metric structure of spacetime from the fundamental ingredients alone. In an upcoming paper \cite{qbit} some of the present authors will show that we can indeed strip away most of the structure assumed here and reconstruct it instead from more fundamental assumptions. 

From the thermodynamic viewpoint, the present paper raises several interesting questions. One is whether we can gain a deeper understanding of the formulation of GR with respect to conjugate variables, cf. \cite{parattu2013structure,padmanabhan2014general}, where a vector field $\xi$ plays an important role. It is tempting to identify this vector field as the one derived from the state on the light cone. This is connected to another question we have left unanswered in the present work, namely how to connect $dP(x, n^\mu)$ to some notion of entropy of spacetime. 

As a follow-up the same question can be asked in the context of Padmanabhan's holographic equipartition in static spacetimes \cite{padmanabhan2012emergence,padmanabhan2014general,padmanabhan2010surface,padmanabhan2010equipartition,padmanabhan2004entropy}. Here, two aspects seem particularly interesting. First, his argument only works for a particular vector field, i.e., in the present language a particular class of measures. Second, for static spacetimes we can make a global split into positive and negative frequencies for the Dirac equation. For CFS this implies that we can define the Dirac sea globally. If, as suggested by the derivation of the Dirac equation in Minkowski space \cite{cfs}, we assume the gravitational degrees of freedom to be in the Dirac sea, then in this setting we get a clean split between the matter degrees of freedom and the gravitational ones. On a more speculative note, this line of investigation could suggest that, in this setting, holography only applies to static or stationary spacetimes. 

Furthermore, it would be interesting to see whether the present formalism can be connected to the line of reasoning developed in \cite{requardt2012gravitons,requardt2012observables,requardt2015incompleteness,requardt2023thermal}. Especially the ideas in the most recent paper in this series \cite{requardt2023thermal} might connect to the results some of the present authors are preparing in the aforementioned follow-up paper \cite{qbit}.

Finally, it seems possible to connect the formalism presented here with Verlinde's derivation of Modified Netwonian Dynamics from entropic arguments \cite{verlinde2017emergent, verlinde2011origin}. In particular, suppose that we consider the interacting region of de Sitter space in static coordinates. 
\begin{equation}
    ds^2= \left( 1- \frac{1}{3}\Lambda r^2 \right)dt^2 - \left( 1- \frac{1}{3}\Lambda r^2 \right)^{-1} dr^2 - r^2 (d\theta^2+\sin^2 \theta d\phi^2)
\end{equation}
If we start with a probability measure $P(x,n^\mu)$ on the PLCB such that $\partial_t=\xi=\la n^\mu\ra_P$ at all points $x\in \mathcal{M}$ with $r(x)\leq r_{\mathcal{H}_C}$ where $r_{\mathcal{H}_C}$ is the coordinate radius of the cosmological horizon. That is, the time-like vector field  $\xi$ associated with the measure $dP(x,n^\mu)$ is the hypersurface normal with respect to a $t=const$ surface in that coordinate system. In this setup Padmanabhan's holographic equipartition of states \cite{padmanabhan2012emergence,padmanabhan2014general,padmanabhan2010surface,padmanabhan2010equipartition,padmanabhan2004entropy} applies, which fits with Verlinde's assumption that the number of bulk degrees of freedom is equal to the number of surface degrees of freedom calculated from the cosmological horizon entropy. 

Now if one wants to realize Verlinde's perturbation of de Sitter space by adding a small mass $M$ at the center of the interacting region, one has to look for a suitable perturbation of the measure $dP(x,n^\mu)$. From the calculations in Section \ref{sec:probonlightcone} a variation $\delta P$ corresponds to a variation of the time-like vector field associated with it
\begin{equation}
    \la n^\mu \ra_{\Lcal(P+\delta P)} = \xi +\delta u.
\end{equation}
Given that $\la n^\mu \ra_{\Lcal(P)} $ is a hypersurface normal vector field, its variation $\delta u$ will necessarily be tangential to the hypersurface $t=const$ for any possible variation. Choosing $\delta P$ to be spherically symmetric around the center of the interacting region seems to be a good starting point for an attempt to reproduce his results from our formalism.
One key step is to relate such a variation $\delta P$ to the addition of a perturbative mass in the center of the interacting region used by Verlinde in his argument. Another question is whether we can associate $\delta P$ with a change in entropy in a suitable sense to match Verlinde's argument. It seems important to note here that Padmanabhan's holographic equipartition breaks down for the perturbed vector field. 

Many of our considerations in this paper were informed by the CFS theory. In particular, the fact that we chose to define the measure $dP(x, n^\mu)$ as a mixed state on the past light cone bundle (PLCB) was motivated by the prominent role that the Dirac sea plays in CFS. Given the emergent nature of spacetime in CFS it is natural to think about a thermodynamic interpretation for the theory. A first step in that direction will be taken in our forthcoming comparison paper between the two approaches \cite{cfstdqg} where the focus will be on the role of a minimal length scale \cite{padmanabhan2020probing}. Another approach is to try and interpret the weight of the measure $P(x,\lambda,\kbar)$ as a weight in the symbol of the fermionic projector. Dropping the spacetime dependence, this motivated the particular choices of $P(\lambda,\kbar)$ in Appendix \ref{sec:regularization} where $P_1$ corresponds to the so-called \textit{hard cut off} regularization and $P_2$ corresponds to the \textit{$i\varepsilon$} regularization. One could then try to interpret the fermionic projector as the Wigner-Weyl quantization cf.\cite{derezinski2020pseudodifferential} of the density of states and relate it to the ideas in \cite{padmanabhan2020geodesic} concerning spacetime correlations. 

To complete the Boltzmannian picture, the missing piece in the present work is an evolution equation for the measure $P(x,\lambda,\kbar)$. Taking again inspiration from CFS where it was shown in \cite{cfstransporteq} that the regularization satisfies a transport equation along null geodesics, this would be the natural first candidate. Similarly, one could postulate for $P(x,\lambda,\kbar)$ to satisfy the massless Vlasov equation.

Finally, there are many choices in this paper that one might make differently. For example, given that our framework is defined on the PLCB, instead of the action principles \eqref{eq:mmtaction} and \eqref{eq:unimodular} one could try building an action involving expressions of the form of the left-hand side of \eqref{eq:jacobson} or \eqref{eq:money} depending explicitly on the null vectors $n^\mu$. Given that such a variation principle contains terms of the form $n^\mu n^\nu$, it could give rise to equations of motion quadratic in the vector fields as they appear in \cite{hossenfelder2017covariant} for a covariant version of Verlinde’s emergent gravity.

\appendix

\section{Positivity Lemma } \label{sec:lemma}
First we will prove a the following Lemma
\begin{Lemma}\label{lem:positivity}
    Let $P(x,y)$ be positive and symmetric, then 
    \begin{equation}
        \int_{R_+\times R_+}  x (x-y) P(x,y) dx dy \geq0
    \end{equation}
\end{Lemma}
\begin{proof}
We simply split the integral along the diagonal
\begin{align}
      &\int_{R_+\times R_+}  x (x-y) P(x,y) dx d y \\
   &= \int_0^\infty \int_y^\infty x (x-y) P(x,y)   dx d y+\int_0^\infty \int^\infty_x x (x-y) P(x,y)  d y dx\\
   &=\int_0^\infty\int_0^\infty(l+y)lP(l+y,y)dy dl +\int_0^\infty\int_0^\infty x(-k)P(x,x+k)dx dk\label{eq:changeofvar}\\
   &=\int_0^\infty\int_0^\infty \left[(l+y)lP(l+y,y)-l y P(l,l+y)\right]dy dl \label{eq:relabel}\\
   &=\int_0^\infty\int_0^\infty l^2P(l+y,y)dy dl\ge0
\end{align}

Where in \eqref{eq:changeofvar} we set $l=x-y$ in the first integral and $k=y-x$ in the second integral. In \eqref{eq:relabel} we simply relabel the variables in the second integral $x=l$ and $k=y$ and sum the integrals.

\end{proof}

\section{Average Vector Field for Specific Probability Distributions}\label{sec:regularization}

In the following, we will calculate $\lambda_{\text{avg}}$ for three particularly relevant probability measures characterized by their weight function with respect to a particular choice of tetrad $P_0=\frac{1}{4\pi}\delta (\varepsilon^{-1}-\lambda)$, $P_1= \frac{\varepsilon}{4\pi}\theta(\lambda)\theta(\varepsilon^{-1} -\lambda)$, and $P_2= \frac{\varepsilon}{4\pi} e^{-\varepsilon \lambda} $. Here, $\theta(\lambda)$ corresponds to the Heaviside function and $P_1$ accordingly to the cutoff regularization in CFS. $P_2$ on the other hand corresponds to the $i\varepsilon$ regularization. 

It is clear that for $i=\{0,1,2\}$
\begin{equation}
    \int_{R_+ \times S^2}  P_i(\lambda,\bar k) d\Omega d \lambda=1
\end{equation}
holds.
We get for $P_0$
\begin{align}
    \lambda_{avg}&= \int_{R_+ \times S^2} \lambda  P_0(\lambda,\bar k) d\Omega d \lambda\\
    &= \epsilon^{-1}
\end{align}
Note that this measure is degenerate of order one, as the zero-zero component of the variance vanishes. Such states are of interest in our follow-up paper \cite{qbit} as they can be related to Hartle-Hawking's no-boundary proposal \cite{hartle1983wave}. 

We get for $P_1$
\begin{align}
    \lambda_{avg}&= \int_{R_+ \times S^2} \lambda  P_1(\lambda,\bar k) d\Omega d \lambda\\
    &=\int_0^{\varepsilon^{-1}}\varepsilon\lambda d\lambda \\
    &= \left[  \frac{\varepsilon \lambda^2}{2} \right]_0^{\varepsilon^{-1}} \\
    &= \frac{\varepsilon^{-1}}{2}. 
\end{align}
We get for $P_2$
\begin{align}
    \lambda_{avg}&= \int_{R_+ \times S^2} \lambda  P_2(\lambda,\bar k) d\Omega d \lambda\\
     &=\int_0^\infty\varepsilon \lambda  e^{-\varepsilon\lambda} d\lambda \\
     &= - \left[\lambda e^{-\varepsilon\lambda} \right]_0^\infty + \int_0^\infty   e^{-\varepsilon\lambda} d\lambda \\
     &= - \left[ \varepsilon^{-1}  e^{-\varepsilon\lambda}   \right]_0^\infty \\
     &= \varepsilon^{-1}
\end{align}
Therefore all three states are naturally associated with a  vector field $\varepsilon^{-1}e_0^\mu$ which we can identify with the regularizing vector field of the locally rigid regularization used for the baryogenesis result. There is a certain constant rescaling of $\lambda_{avg}$ between different regularizations that might matter in detailed phenomenological calculations. Over all this gives a tempting connection to CFS.

Note that it is a curious observation that due to the $S^2$ factor in the light cone we always need to normalize by a factor of $4\pi$, which Padmanabhan needs to determine the current value of the cosmological constant. This need not mean much, as it will obviously show up in any calculation that involves $S^2$, however, due to the fact that we demonstrated in \eqref{eq:money} and \eqref{eq:moneynormalized} that there is a direct connection between $dP$ and the cosmological constant, there might be something more to it. 

Note that if we make any of these states dependent on $x$ by $P_0=\frac{1}{4\pi}\delta (\varepsilon(x)^{-1}-\lambda)$, $P_1= \frac{\varepsilon_0}{4\pi}\theta(\lambda)\theta(\varepsilon(x)^{-1} -\lambda)$ and $P_2= \frac{\varepsilon_0}{4\pi} e^{-\varepsilon (x)\lambda} $ then for $P_1$ and $P_2$ we immediately get $\nabla_\mu\chi(x)\neq 0$ when we allow the regularization length $\varepsilon(x)$ to vary, while for $P_0$ we get $\chi(x)=\text{const}$.

\bibliographystyle{amsplain}
\bibliography{claudio}

\providecommand{\bysame}{\leavevmode\hbox to3em{\hrulefill}\thinspace}
\providecommand{\MR}{\relax\ifhmode\unskip\space\fi MR }
\providecommand{\MRhref}[2]{%
  \href{http://www.ams.org/mathscinet-getitem?mr=#1}{#2}
}
\providecommand{\href}[2]{#2}
\begin{thebibliography}{10}

\bibitem{website}
\emph{\emph{Link to web platform on causal fermion systems:
  www.causal-fermion-system.com}}.

\bibitem{alonso2022thermodynamics}
A.~Alonso-Serrano and M.~Li{\v{s}}ka, \emph{Thermodynamics of spacetime and
  unimodular gravity}, International Journal of Geometric Methods in Modern
  Physics (2022), 2230002.

\bibitem{banihashemi2022thermodynamic}
B.~Banihashemi and T.~Jacobson, \emph{Thermodynamic ensembles with cosmological
  horizons}, Journal of High Energy Physics \textbf{2022} (2022), no.~7, 1--59.

\bibitem{bardeen1973four}
J.~M. Bardeen, B.~Carter, and S.~W. Hawking, \emph{The four laws of black hole
  mechanics}, Communications in Mathematical Physics \textbf{31} (1973),
  161--170.

\bibitem{bekenstein1973black}
J.~D. Bekenstein, \emph{Black holes and entropy}, Physical Review D \textbf{7}
  (1973), no.~8, 2333.

\bibitem{bekenstein1980black}
\bysame, \emph{Black-hole thermodynamics}, Physics Today \textbf{33} (1980),
  no.~1, 24--31.

\bibitem{Bruce_2022}
A.~J. Bruce, \emph{On the bundle of null cones}, International Journal of
  Geometric Methods in Modern Physics \textbf{20} (2022), no.~02.

\bibitem{bueno2017entanglement}
P.~Bueno, V.~S. Min, A.~J. Speranza, and M.~R. Visser, \emph{Entanglement
  equilibrium for higher order gravity}, Physical Review D \textbf{95} (2017),
  no.~4, 046003.

\bibitem{carballo2022unimodular}
R.~Carballo-Rubio, L.~J. Garay, and G.~Garc{\'\i}a-Moreno, \emph{Unimodular
  gravity vs general relativity: a status report}, Classical and Quantum
  Gravity \textbf{39} (2022), no.~24, 243001.

\bibitem{castrigiano2017dirac}
D.~P.~L. Castrigiano, \emph{Dirac and weyl fermions--the only causal systems},
  arXiv:1711.06556 (2017).

\bibitem{chakraborty2019raychaudhuri}
S.~Chakraborty, D.~Kothawala, and A.~Pesci, \emph{Raychaudhuri equation with
  zero point length}, Physics Letters B \textbf{797} (2019), 134877.

\bibitem{chernov2010legendrian}
V.~Chernov and S.~Nemirovski, \emph{Legendrian links, causality, and the low
  conjecture}, Geometric and Functional Analysis \textbf{19} (2010), no.~5,
  1320--1333.

\bibitem{chirco2010nonequilibrium}
G.~Chirco and S.~Liberati, \emph{Nonequilibrium thermodynamics of spacetime:
  The role of gravitational dissipation}, Physical Review D \textbf{81} (2010),
  no.~2, 024016.

\bibitem{derezinski2020pseudodifferential}
J.~Derezi{\'n}ski, A.~Latosi{\'n}ski, and D.~Siemssen, \emph{Pseudodifferential
  weyl calculus on (pseudo-) riemannian manifolds}, Annales Henri Poincar{\'e},
  vol.~21, Springer, 2020, pp.~1595--1635.

\bibitem{ehlers2012republication}
J.~Ehlers, F.~A.~E. Pirani, and A.~Schild, \emph{Republication of: The geometry
  of free fall and light propagation}, General Relativity and Gravitation
  \textbf{44} (2012), no.~6, 1587--1609.

\bibitem{cfs}
F.~Finster, \emph{The {C}ontinuum {L}imit of {C}ausal {F}ermion {S}ystems},
  arXiv:1605.04742, Fundamental Theories of Physics, vol. 186, Springer, 2016.

\bibitem{ethcfs}
F.~Finster, J.~Fr{\"o}hlich, M.~Oppio, and C.~F Paganini, \emph{Causal fermion
  systems and the eth approach to quantum theory}, Discrete \& Continuous
  Dynamical Systems-S \textbf{14} (2021), no.~5, 1717.

\bibitem{oct-cfs}
F.~Finster, N.~G. Gresnigt, J.~M. Isidro, A.~Marciano, C.~F. Paganini, and
  T.~P. Singh, \emph{Causal fermion systems and octonions}, arXiv:2403.00360
  (2024).

\bibitem{mmt-cfs}
F.~Finster, E.~Guendelman, and C.~Paganini, \emph{Modified measures as an
  effective theory for causal fermion systems}, arXiv:2303.16566, Classical and
  Quantum Gravity \textbf{41} (2024), no.~3, 035007, 25.

\bibitem{cfstdqg}
F.~Finster, J.~M. Isidro, and C.~F. Paganini, \emph{Causal fermion systems as
  underlying theory for a thermodynamic interpretation of gravity}, in
  preparation.

\bibitem{tracedynamics}
F.~Finster, J.M. Isidro, C.F. Paganini, and T.P. Singh, \emph{Causal fermion
  systems, trace dynamics and the spectral action principle}, in preparation.

\bibitem{baryogenesis}
F.~Finster, M.~Jokel, and C.~F. Paganini, \emph{A mechanism of baryogenesis for
  causal fermion systems}, arXiv:2111.05556, Classical and Quantum Gravity
  (2021).

\bibitem{cfstransporteq}
F.~Finster and M.~Kraus, \emph{The regularized hadamard expansion}, Journal of
  Mathematical Analysis and Applications \textbf{491} (2020), no.~2, 124340.

\bibitem{finster2023incompatibility}
F.~Finster and C.~F. Paganini, \emph{Incompatibility of frequency splitting and
  spatial localization: a quantitative analysis of hegerfeldt’s theorem},
  Annales Henri Poincar{\'e}, vol.~24, Springer, 2023, pp.~413--467.

\bibitem{gibbons1977action}
G.~W. Gibbons and S.~W. Hawking, \emph{Action integrals and partition functions
  in quantum gravity}, Physical Review D \textbf{15} (1977), no.~10, 2752.

\bibitem{Fieltheory}
E.~Guendelman and A.Kaganovich, \emph{Dynamical measure and field theory models
  free of the cosmological constant problem}, arXiv:gr-qc/9905029, Physical
  Review D \textbf{60} (1999), 065004.

\bibitem{manymeasures}
E.~I. Guendelman, R.~Herrera, P.~Labrana, E.~Nissimov, and S.~Pacheva,
  \emph{Emergent cosmology, inflation and dark energy}, arXiv:1408.5344,
  General Relativity and Gravitation \textbf{47} (2015), 2, 10.

\bibitem{NGVE}
E.~I. Guendelman and A.~Kaganovich, \emph{The principle of nongravitating
  vacuum energy and some of its consequences}, arXiv:gr-qc/9605026, Physical
  Review D \textbf{53} (1996), 7020--7025.

\bibitem{hartle1983wave}
J.~B. Hartle and S.~W. Hawking, \emph{Wave function of the universe}, Physical
  Review D \textbf{28} (1983), no.~12, 2960.

\bibitem{hawking1974black}
S.~W. Hawking, \emph{Black hole explosions?}, Nature \textbf{248} (1974),
  no.~5443, 30--31.

\bibitem{hawking1975particle}
\bysame, \emph{Particle creation by black holes}, Communications in
  Mathematical Physics \textbf{43} (1975), no.~3, 199--220.

\bibitem{hawking1976new}
S.~W. Hawking, A.~R. King, and P.~J. McCarthy, \emph{A new topology for curved
  space--time which incorporates the causal, differential, and conformal
  structures}, Journal of Mathematical Physics \textbf{17} (1976), no.~2,
  174--181.

\bibitem{hedicke2020conformally}
J.~Hedicke and S.~Suhr, \emph{Conformally embedded spacetimes and the space of
  null geodesics}, Communications in Mathematical Physics \textbf{375} (2020),
  no.~2, 1561--1577.

\bibitem{hegerfeldt1974remark}
G.~C. Hegerfeldt, \emph{Remark on causality and particle localization},
  Physical Review D \textbf{10} (1974), no.~10, 3320.

\bibitem{UnimodularGRI}
M.~Henneaux and C.~Teitelboim, \emph{The cosmological constant and general
  covariance}, Physics Letters B \textbf{222} (1989), 195--199.

\bibitem{hossenfelder2013minimal}
S.~Hossenfelder, \emph{Minimal length scale scenarios for quantum gravity},
  Living Reviews in Relativity \textbf{16} (2013), no.~1, 2.

\bibitem{hossenfelder2017covariant}
\bysame, \emph{Covariant version of verlinde’s emergent gravity}, Physical
  Review D \textbf{95} (2017), no.~12, 124018.

\bibitem{jacobson1995thermodynamics}
T.~Jacobson, \emph{Thermodynamics of spacetime: the einstein equation of
  state}, Physical Review Letters \textbf{75} (1995), no.~7, 1260.

\bibitem{jacobson2008einstein}
\bysame, \emph{Einstein-aether gravity: A status report}, arXiv:0801.1547
  (2008).

\bibitem{jacobson2016entanglement}
\bysame, \emph{Entanglement equilibrium and the {E}instein equation}, Physical
  Review Letters \textbf{116} (2016), no.~20, 201101.

\bibitem{jacobson2023partition}
T.~Jacobson and M.~R. Visser, \emph{Partition function for a volume of space},
  Physical Review Letters \textbf{130} (2023), no.~22, 221501.

\bibitem{jain2012testing}
P.~Jain, P.~Karmakar, S.~Mitra, S.~Panda, and N.~K. Singh, \emph{Testing
  unimodular gravity}, Journal of Cosmology and Astroparticle Physics
  \textbf{2012} (2012), no.~05, 020.

\bibitem{kothawala2013minimal}
D.~Kothawala, \emph{Minimal length and small scale structure of spacetime},
  Physical Review D \textbf{88} (2013), no.~10, 104029.

\bibitem{nonlocal2023kothawala}
\bysame, \emph{Limits of a non-local quantum spacetime}, International Journal
  of Modern Physics D \textbf{32} (2023), no.~14, 2342021.

\bibitem{kothawala2014entropy}
D.~Kothawala and T.~Padmanabhan, \emph{Entropy density of spacetime as a relic
  from quantum gravity}, Physical Review D \textbf{90} (2014), no.~12, 124060.

\bibitem{linnemann2021constructive}
N.~Linnemann and J.~Read, \emph{Constructive axiomatics in spacetime physics
  part i: Walkthrough to the ehlers-pirani-schild axiomatisation},
  arXiv:2112.14063 (2021).

\bibitem{low1989geometry}
R.~J. Low, \emph{The geometry of the space of null geodesics}, Journal of
  mathematical physics \textbf{30} (1989), no.~4, 809--811.

\bibitem{malament1977class}
D.~B. Malament, \emph{The class of continuous timelike curves determines the
  topology of spacetime}, Journal of Mathematical Physics \textbf{18} (1977),
  no.~7, 1399--1404.

\bibitem{minguzzi2007causal}
E.~Minguzzi, \emph{The causal ladder and the strength of k-causality: I},
  Classical and Quantum Gravity \textbf{25} (2007), no.~1, 015009.

\bibitem{minguzzi2007causal2}
\bysame, \emph{The causal ladder and the strength of k-causality: Ii},
  Classical and Quantum Gravity \textbf{25} (2007), no.~1, 015010.

\bibitem{ng1991unimodular}
Y.~J. Ng and H.~Van~Dam, \emph{Unimodular theory of gravity and the
  cosmological constant}, Journal of Mathematical Physics \textbf{32} (1991),
  no.~5, 1337--1340.

\bibitem{padilla2015note}
A.~Padilla and I.~D. Saltas, \emph{A note on classical and quantum unimodular
  gravity}, The European Physical Journal C \textbf{75} (2015), 1--5.

\bibitem{padmanabhan2004entropy}
T.~Padmanabhan, \emph{Entropy of static spacetimes and microscopic density of
  states}, Classical and Quantum Gravity \textbf{21} (2004), no.~18, 4485.

\bibitem{padmanabhan2009entropy}
\bysame, \emph{Entropy density of spacetime and thermodynamic interpretation of
  field equations of gravity in any diffeomorphism invariant theory},
  arXiv:0903.1254 (2009).

\bibitem{padmanabhan2010equipartition}
\bysame, \emph{Equipartition of energy in the horizon degrees of freedom and
  the emergence of gravity}, Modern Physics Letters A \textbf{25} (2010),
  no.~14, 1129--1136.

\bibitem{padmanabhan2010surface}
\bysame, \emph{Surface density of spacetime degrees of freedom from
  equipartition law in theories of gravity}, Physical Review D \textbf{81}
  (2010), no.~12, 124040.

\bibitem{padmanabhan2011entropy}
\bysame, \emph{Entropy density of spacetime and the navier-stokes fluid
  dynamics of null surfaces}, Physical Review D \textbf{83} (2011), no.~4,
  044048.

\bibitem{padmanabhan2012emergence}
\bysame, \emph{Emergence and expansion of cosmic space as due to the quest for
  holographic equipartition}, arXiv:1206.4916 (2012).

\bibitem{padmanabhan2014general}
\bysame, \emph{General relativity from a thermodynamic perspective}, General
  Relativity and Gravitation \textbf{46} (2014), no.~3, 1673.

\bibitem{padmanabhan2015one}
\bysame, \emph{One hundred years of general relativity: Summary, status and
  prospects}, Current Science \textbf{109} (2015), no.~7, 1215--1219.

\bibitem{padmanabhan2016atoms}
\bysame, \emph{The atoms of space, gravity and the cosmological constant},
  International Journal of Modern Physics D \textbf{25} (2016), no.~07,
  1630020.

\bibitem{padmanabhan2017atoms}
\bysame, \emph{The atoms of spacetime and the cosmological constant}, Journal
  of Physics: Conference Series, vol. 880, IOP Publishing, 2017, p.~012008.

\bibitem{padmanabhan2018kinetic}
\bysame, \emph{The kinetic theory of the mesoscopic spacetime}, International
  Journal of Modern Physics D \textbf{27} (2018), no.~14, 1846004.

\bibitem{padmanabhan2020geodesic}
\bysame, \emph{Geodesic distance: A descriptor of geometry and correlator of
  pregeometric density of spacetime events}, Modern Physics Letters A
  \textbf{35} (2020), no.~12, 2030008.

\bibitem{padmanabhan2020probing}
\bysame, \emph{Probing the planck scale: the modification of the time evolution
  operator due to the quantum structure of spacetime}, Journal of High Energy
  Physics \textbf{2020} (2020), no.~11, 1--26.

\bibitem{padmanabhan2022microscopic}
T.~Padmanabhan and S.~Chakraborty, \emph{Microscopic origin of einstein's field
  equations and the raison d'{\^e}tre for a positive cosmological constant},
  Physics Letters B \textbf{824} (2022), 136828.

\bibitem{padmanabhan2016spacetime}
T.~Padmanabhan, S.~Chakraborty, and D.~Kothawala, \emph{Spacetime with zero
  point length is two-dimensional at the planck scale}, General Relativity and
  Gravitation \textbf{48} (2016), 1--8.

\bibitem{padmanabhan2017cosmic}
T.~Padmanabhan and H.~Padmanabhan, \emph{Cosmic information, the cosmological
  constant and the amplitude of primordial perturbations}, Physics Letters B
  \textbf{773} (2017), 81--85.

\bibitem{qbit}
C.~F. Paganini and M.~van~den Beld-Serrano, \emph{Constructing a lorentzian
  metric from a family of probability fields on a sphere bundle}, in
  preparation.

\bibitem{parattu2013structure}
K.~Parattu, B.~R. Majhi, and T.~Padmanabhan, \emph{Structure of the
  gravitational action and its relation with horizon thermodynamics and
  emergent gravity paradigm}, Physical Review D \textbf{87} (2013), no.~12,
  124011.

\bibitem{parikh2018einstein}
M.~Parikh and A.~Svesko, \emph{Einstein’s equations from the stretched future
  light cone}, Physical Review D \textbf{98} (2018), no.~2, 026018.

\bibitem{pesci2019quantum}
A.~Pesci, \emph{Quantum metric for null separated events and spacetime atoms},
  Classical and Quantum Gravity \textbf{36} (2019), no.~7, 075009.

\bibitem{pesci2019spacetime}
\bysame, \emph{Spacetime atoms and extrinsic curvature of equi-geodesic
  surfaces}, The European Physical Journal Plus \textbf{134} (2019), no.~7,
  374.

\bibitem{pesci2020minimum}
\bysame, \emph{Minimum-length ricci scalar for null separated events}, Physical
  Review D \textbf{102} (2020), no.~12, 124057.

\bibitem{requardt2012gravitons}
M.~Requardt, \emph{Gravitons as goldstone modes and the spontaneous symmetry
  breaking of diffeomorphism invariance}, arXiv:1203.1702 (2012).

\bibitem{requardt2012observables}
\bysame, \emph{Observables need not be diffeomorphism invariant in classical
  and quantum gravity}, arXiv:1206.0832 (2012).

\bibitem{requardt2015incompleteness}
\bysame, \emph{About the incompleteness of the semiclassical picture in the
  bh-information paradox and the quantum hair emerging from the coupling of
  quantum matter and quantum space-time degrees of freedom}, arXiv:1503.07312
  (2015).

\bibitem{requardt2023thermal}
\bysame, \emph{The thermal substructure of general relativity},
  arXiv:2301.00980 (2023).

\bibitem{rovelli2019travel}
C.~Rovelli, \emph{{Can we travel to the past? Irreversible physics along closed
  timelike curves}}, 2019.

\bibitem{singh2021quantum}
T.~P. Singh, \emph{Quantum gravity, minimum length and holography}, Pramana
  \textbf{95} (2021), 1--7.

\bibitem{smolin2009quantization}
L.~Smolin, \emph{Quantization of unimodular gravity and the cosmological
  constant problems}, Physical Review D \textbf{80} (2009), no.~8, 084003.

\bibitem{stargen2015small}
D.~J. Stargen and D.~Kothawala, \emph{Small scale structure of spacetime: The
  van vleck determinant and equigeodesic surfaces}, Physical Review D
  \textbf{92} (2015), no.~2, 024046.

\bibitem{svesko2019entanglement}
A.~Svesko, \emph{From entanglement to thermodynamics and to gravity}, Physical
  Review D \textbf{99} (2019), no.~8, 086006.

\bibitem{synge1960relativity}
J.~L. Synge, \emph{Relativity: The general theory}, Co., Amsterdam (1960), 450.

\bibitem{verlinde2011origin}
E.~Verlinde, \emph{On the origin of gravity and the laws of newton}, Journal of
  High Energy Physics \textbf{2011} (2011), no.~4, 1--27.

\bibitem{verlinde2017emergent}
E.~P. Verlinde, \emph{Emergent gravity and the dark universe}, SciPost Physics
  \textbf{2} (2017), no.~3, 016.

\bibitem{wald1999gravitation}
R.~M. Wald, \emph{Gravitation, thermodynamics and quantum theory}, Classical
  and Quantum Gravity \textbf{16} (1999), no.~12A, A177.

\end{thebibliography}

\end{document}